\documentclass[draft]{amsart}
\newtheorem{thm}{Theorem}[section]

\newtheorem{defn}[thm]{Definition}

\usepackage{amsfonts}
\def\C{{\mathbb C}}
\def\sCM#1#2{\hbox{sCM}^{#1}_{#2}}
\def\eqref#1{$(\ref{#1})$}
\def\GL{\textup{GL}}
\def\V{\nu}
\def\t{\vec t}
\def\exp{\hbox{exp}}
\def\rank{\textup{rank}\,}
\def\Res#1{\mathop{\textup{Res}}_{z=#1}}
\def\calL{{\mathcal{L}}}
\def\b{\flat}
\def\N{V}
\def\R{\mathcal{R}}
\def\K{\mathcal{K}}
\def\p{{q}}
\def\Left#1{{#1}_{\hbox{\scriptsize\tt L}}}
\def\Right#1{{#1}_{\hbox{\scriptsize\tt R}}}
\def\cc(#1){c_{#1}}

\newcounter{claimct}[section]
\def\theclaimct{\arabic{section}.\arabic{claimct}}
\def\claim#1{\par\medskip\par\noindent\refstepcounter{claimct}\textbf{#1 \theclaimct:}}

\title[Spin Calogero and Bispectral Matrix KP]{Spin Calogero
  Particles and Bispectral Solutions of the Matrix KP Hierarchy}

\author{Maarten Bergvelt}
    \address[Bergvelt] {Department of Mathematics,
University of Illinois, Urbana, IL 61801, USA }
    \email{bergv@uiuc.edu}
    \author{Michael Gekhtman}
    \address[Gekhtman] {Department of Mathematics,
University of Notre Dame, Notre Dame, IN 46556, USA }
    \email{gekhtman.1@nd.edu}

 \author{Alex Kasman}
    \address[Kasman] {Department of Mathematics,
College of Charleston, Char\-le\-ston, SC 29424, USA }
    \email{kasmana@cofc.edu}

\date{\today}

\begin{document}
\begin{abstract}
  Pairs of $n\times n$ matrices whose commutator differ from the
  identity by a matrix of rank $r$ are used to construct bispectral
  differential operators with $r\times r$ matrix coefficients
  satisfying the Lax equations of the Matrix KP hierarchy.  Moreover,
  the bispectral involution on these operators has dynamical
  significance for the spin Calogero particles system whose phase
  space such pairs represent.  In the case $r=1$, this reproduces
  well-known results of Wilson and others from the 1990's relating
  (spinless) Calogero-Moser systems to the bispectrality of (scalar)
  differential operators. This new class of pairs $(L, \Lambda)$ of
  bispectral matrix differential operators is different than those
  previously studied in that $L$ acts from the left, but $\Lambda$
  from the right on a common $r\times r$ eigenmatrix.
\end{abstract}

\maketitle

\section{Introduction}

\subsection{Background}
Let 
\begin{equation}
\hbox{CM}_n=\{(X,Z)\in M_{n\times n}\,|\,\rank([X,Z]-I)=1\}\label{rankone}
\end{equation}
be the set of pairs of complex $n\times n$ matrices whose commutator
differs from the identity by a matrix of rank one.  This space arises
naturally in the study of the integrable Calogero-Moser-Sutherland
particle system \cite{MR1843558, MR2277458}.  In particular, the
eigenvalues of the time dependent matrix $X+i t Z^{i-1}$ move according to the $i$th
Hamiltonian of this integrable hierarchy and even allows the
continuation of the dynamics through collisions \cite{MR0478225, MR1626461}.

The KP hierarchy  is the collection of nonlinear partial differential equations 
\begin{equation}
\label{KPlax}
\frac{\partial}{\partial t_i} \calL=[(\calL^i)_+,\calL],\qquad i=1,2,3.\ldots
\end{equation}
for a monic pseudo-differential operator $\calL$ of order one whose
coefficients are scalar functions depending on the time variables
$t_i$ \cite{MR86m:58072,MR87b:58039}.  If the coefficients of $\calL$
are further assumed to be rational functions of $t_1$ which vanish as
$t_1\to\infty$, then the solutions can be written in terms of the
matrices in $\hbox{CM}_{n}$ and the poles move according to the
dynamics of the Calogero-Moser-Sutherland system
\cite{MR0649926,MR557031,MR1299922,MR1626461}.  This was interpreted
as a special case of a more general relationship between ``rank one
conditions'' and the KP hierarchy in \cite{MR2173898}.

Although it seems at first to be quite different in nature, having no
obvious dynamical interpretation, the bispectral problem
\cite{MR1611018} turns out to be another aspect of this relationship
between the KP hierarchy and the Calogero-Moser-Sutherland particle
system.  As originally formulated in \cite{MR826863}, the bispectral problem
seeks to find scalar coefficient ordinary differential operators $L$
and $\Lambda$ in the variables $x$ and $z$ respectively such that
there is a common eigenfunction $\psi(x,z)$ satisfying the eigenvalue
equations
\begin{equation}
\label{bispdef1}
L\psi =p(z)\psi \qquad \Lambda\psi=\pi(x)\psi
\end{equation}
for non-constant functions $p$ and $\pi$.  As it turns out, if one
additionally requires the operator $L$ to commute with another
ordinary differential operator of relatively prime order, the
solutions to the bispectral problem are exactly the same as the
rational solutions to the KP hierarchy mentioned above \cite{MR1234841}.
(Specifically, up to trivial renormalizations, the bispectral
operators are the ordinary differential operators that commute with
the pseudo-differential operator $\calL$ with the identification
$x=t_1$.)  Moreover, the bispectral property for these operators is a
manifestation of the involution on $\hbox{CM}_n$ given by
$(X,Z)\mapsto(Z^{\top},X^{\top})$ which linearizes the dynamics of the
particle system \cite{MR1350415,MR1626461}.

In \cite{MR1626461} and the conference proceedings \cite{MR1844234},
Wilson suggests that the correspondence should generalize naturally to
the case in which the ``rank one condition'' \eqref{rankone} is
replaced with a ``rank $r$ condition''.  Indeed, various authors have
demonstrated that a similar relationship exists between the matrices
whose commutator differs from the identity by a matrix of rank $r$
($1< r \leq n$), the ``spin generalization'' of the integrable
particle system, and matrix generalizations of the KP hierarchy.  In
particular, the spin generalized system \cite{MR761664} was shown to be
related to the matrix KP equation in \cite{MR1355552} and to the
multi-component KP hierarchy in \cite{MR2377220}, and rather general
rank $r$ conditions were shown to produce solutions to the matrix
potential KP hierarchy in \cite{MR2321659}.  None of these, however, has
specifically addressed the question of whether and how the results
relating to bispectrality generalize to the matrix case.

\subsection{Outline}

Section~\ref{sec:notation} will introduce a version of the bispectral
problem in which matrix differential operators act on a common
eigenmatrix from opposite
sides.  The main result of this paper will be to demonstrate that this
formulation allows for the generalization of the results on
bispectrality to the spin version of the particle system and the
matrix KP hierarchy.

Section~\ref{sec:sCM} introduces the generalization of \eqref{rankone}
to the case of arbitrary rank and relates it to the dynamics of the
spin Calogero particle system.  Special attention is paid to the block
decompositions of the associated operators corresponding to the
generalized eigenspaces of the matrix $Z$.  

A wave function and pseudo-differential operator are constructed from
a choice of $n\times n$ matrices satisfying the rank $r$ condition in
Section~\ref{sec:MKP}.  This $r\times r$ matrix pseudo-differential
operator is shown to satisfy the Lax equation of the KP hierarchy
\eqref{KPlax}.  A key component of the proof is the explicit
construction of an $rn$-dimensional space of finitely supported
distributions in the spectral parameter which annihilate the wave
function.

Several obvious group actions on the space of matrices satisfying the
rank $r$ condition are investigated in Section~\ref{sec:symmetries}
with emphasis on their effect on the corresponding KP solution.  Of
special interest is the bispectral involution which has the effect of
exchanging variables and transposing the wave function.

The main theorem is the construction in
Section~\ref{sec:bispectrality} of commutative rings of matrix
differential operators in $x$ and $z$ and the demonstration that they
have the KP wave function as a common eigenfunction.

A final section contains closing remarks and lists problems for future research on this topic.

 \subsection{Notation and Matrix Bispectrality}\label{sec:notation}
 
 We will make use of the notation $\Left{M}$ and $\Right{M}$ to
 distinguish between the cases in which the operator $M$ is acting
 from the left or the right, respectively.  So, for instance, if $M$
 and $P$ are both $r\times r$ matrices, then
$$
[M,P]=(\Left{M}-\Right{M})P.
$$

Similarly, if
\begin{equation}
L=\sum_{i=0}^N M_i(x) \partial_x^i\label{eq:1}
\end{equation}
is an ordinary differential operator in $x$ of degree $N$ with
coefficients $M_i(x)$ that are $r\times r$ matrices and $\psi(x)$ is
an $r\times r$ matrix function we define
$$
\Left{L}(\psi)=L(\psi)=\sum_{i=0}^N M_i(x)\left(\frac{\partial^i}{\partial x^i} \psi(x)\right),
$$
as usual.   However, the operator can also act from the right 
$$
\Right{L}(\psi)=\sum_{i=0}^N \left(\frac{\partial^i}{\partial x^i} \psi(x)\right)M_i(x).
$$
Equivalently, if we denote by $L^{\top}$ the differential operator
with coefficients $M_i^{\top}$ that are the
ordinary matrix transpose of the coefficients of $L$, we can say
\begin{equation}
\Right{L}(\psi)=\left(L^{\top}(\psi^{\top})\right)^{\top}.\label{eq:RightTranspose}
\end{equation}

\begin{defn} \label{def:bisp}
  A \emph{bispectral triple} $(L,\Lambda,\psi)$ consists of a
  differential operator $L$ in $x$ as in (\ref{eq:1}), a differential
  operator $\Lambda$ in the variable $z$ also having $r\times r$
  matrix coefficients, and an
  $r\times r$ matrix function $\psi(x,z)$ of $x$ and $z$ satisfying
  the equations 
  \begin{equation}
\Left{L}(\psi)=p(z)\psi \qquad \hbox{and}\qquad \Right{\Lambda}(\psi)=\pi(x)\psi,
\label{bispdef2}
\end{equation}
where $p(z)$ and $\pi(x)$ are non-constant, scalar eigenvalues.   
\end{defn}

This seems to be a natural matrix generalization of the scalar
bispectral problem for differential operators considered in
\cite{MR826863}.  However, we note that this differs
from the matrix generalization previously considered by Zubelli
\cite{MR1067913} in which both operators acted
from the same side, and also from the ``bundle bispectrality''
considered by Sakhnovich-Zubelli \cite{MR1857803} where the operators
$L$ and $\Lambda$ were allowed to depend on both variables.  (Here we
are interested only in the case that $L$ is independent of $z$ and
$\Lambda$ is independent of $x$.)

In the rest of the paper we will use the notation $I_{k}$ for the
$k\times k$ identitity matrix. Also we will abuse notation by
using the symbol $I$ to denote the
identity transformation on many different vector spaces whenever its
use should make it clear which is intended.
\section{Spin Calogero Matrices}\label{sec:sCM}

Let $\sCM{r}{n}$ be the the set of 4-tuples of matrices $(X,Z,A,B)$ such that the $n\times n$ commutator $[X,Z]$ differs from the identity by the rank $r$ matrix $BA$:
\begin{equation}
  \label{eq:2}
  \sCM{r}{n}=\{(X,Z,A,B)\mid  X,Z\in M_{n\times n},\ A,B^{\top}\in M_{r\times n},\  [X,Z]-I=BA\not=0\}.
\end{equation}
This space arises naturally in the description of generic initial conditions of the Spin Calogero
particles, as we will see below.  More importantly, the
dynamics linearizes there (when one considers the phase space to be
$\sCM{r}{n}$ modulo the action of $\GL(n)$ to be described in the
section on symmetries).

Let $q_i$ ($1\leq i\leq n$) be the distinct positions of $n$ particles
on the complex plane, $\dot q_i$ be their momenta, and
$f_{ij}=\beta_j\alpha_i $ be their ``spins'' represented as the
products of $n$ column $r$-vectors $\alpha_i$ and $n$ row $r$-vectors
$\beta_j$ subject to the constraint $f_{ii}=-1$.

We associate to this data the matrices $(X,Z,A,B)\in \sCM{r}{n}$ in the form
$$
X_{ij}=q_i \delta_{ij} \qquad Z_{ij}=\dot q_i \delta_{ij} + (1-\delta_{ij})\frac{f_{ij}}{q_i-q_j}
$$
$$
A=
\begin{pmatrix}\alpha_1& \cdots & \alpha_n\end{pmatrix}
\qquad
B=\begin{pmatrix}\beta_1\cr
\vdots\cr
\beta_n\end{pmatrix}.
$$
The dynamics of the eigenvalues of $X+ i t
Z^{i-1}$ are governed by the Hamiltonian $H=\textup{tr} Z^i$.  This is
the \textit{spin Calogero system} \cite{MR761664,MR557031}.  In the
special case $r=1$, this reduces to the more famous (spinless)
Calogero-Moser-Sutherland particle system \cite{MR1843558}.

\subsection{Block Decomposition}\label{sec:block}

For a fixed choice of $(X,Z,A,B)\in\sCM{r}{n}$ we get a decomposition
of $V=\C^n$ into generalized eigenspaces of $Z$:
$$
V=\bigoplus_{\lambda} \N_{\lambda} \qquad
\N_{\lambda}=\{v\in\C^n\,|\, (Z-\lambda I)^kv=0\hbox{ for some
}k\geq0\}.
$$
The restriction of $Z$ to $\N_{\lambda}$ will be denoted by
$$
Z_{\lambda}=\lambda I + N_{\lambda}
$$
where $I$ is the identity operator on $\N_{\lambda}$ and $N_{\lambda}$ is nilpotent.
Rational expressions in $zI-Z_{\lambda}$
below will always be interpreted by
expansion in positive powers of the nilpotent $N_{\lambda}$.

We will utilize subscripts $\lambda$ and $\mu$ which will run over the
eigenvalues of $Z$ to similarly denote the blocks of other linear
operators associated to this decomposition of $\C^n$.  Specifically,
$A_{\lambda}:\N_{\lambda}\to \C^r$ will be the restriction of the map
$A$, $B_{\lambda}:\C^r\to \N_{\lambda}$ will be the map $B$ followed by
projection onto $\N_{\lambda}$, and for a linear operator $M$ from $\C^r$
to itself (such as $X$) $M_{\lambda,\mu}$ will be the block
corresponding to the map from $\N_{\mu}$ to $\N_{\lambda}$.  

The sCM
condition (\ref{eq:2}) involves the commutator
$[X,Z]=(\Right{Z}-\Left{Z})(X)$.  Interestingly, although the operator
$\Right{Z}-\Left{Z}=-\textup{ad}(Z)$ is not invertible, its
``off-diagonal'' action is invertible which allows us to solve for
$X_{\lambda\mu}$ when $\mu\not=\lambda$.

\claim{Lemma}\label{lem:offdiagX}
Let $\lambda\not=\mu$ be generalized eigenvalues of $Z$.  Then
$$
X_{\mu\lambda}=\left(\Right{(Z_{\lambda})}-\Left{(Z_{\mu})}\right)^{-1}
B_{\mu}A_{\lambda}
=\sum_{k\geq0}\frac{(\Right{(N_{\lambda})}-\Left{(N_{\mu})})^k}{(\lambda-\mu)^{k+1}}
B_{\mu}A_{\lambda}.
$$
\begin{proof}
Since
$$(\Right{Z}-\Left{Z})_{\mu\lambda}=(\lambda-\mu)I_{\mu\lambda}+\Right{(N_{\lambda})}-\Left{(N_{\mu})}
$$
differs from the a nonzero multiple of the identity by a nilpotent
matrix, it is invertible.  Specifically, we may invert it in general
using
$$
(\Right{Z}-\Left{Z})_{\mu\lambda}^{-1}=\sum_{k\geq0}\frac{(\Right{(N_{\lambda})}-\Left{(N_{\mu})})^k}{(\lambda-\mu)^{k+1}}.
$$
Applying this when solving  $[X,Z]-I=BA$ for any off-diagonal block of $X$ yields the claimed formula.
\end{proof}

It will later be necessary to evaluate residues of matrix functions
written in terms of the blocks $Z_{\lambda}$.  For this purpose the
following ``obvious'' lemma will be useful. For convenience we
introduce notation for ``divided derivatives'':
\[
f^{[k]}=\frac{1}{k!}\frac{d^{k}f}{dz^{k}}.
\]
\claim{Lemma}\label{lem:matres} Let $f(z)$ be a rational function that
is regular at $z=\lambda$, then
\[
\Res{\lambda}\left(\frac{f(z)}{(zI-Z_{\lambda})^{k+1}}\right)=f^{[k]}(Z_{\lambda}).
\]
\begin{proof}
\begin{align*}
\Res{\lambda}&\left(\frac{f(z)}{(zI-Z_{\lambda})^{k+1}}\right)
=\Res{\lambda}\left(f(z)\frac{(-\partial_{z})^{k}}{k!}\frac{1}{zI-Z_{\lambda}}\right)\\
&=\Res{\lambda}\left(f^{[k]}(z)\frac{1}{zI-Z_{\lambda}}\right)
=\Res{\lambda}\left(f^{[k]}(z)\sum_{s\geq 0}\frac{N_{\lambda}^s}{(z-\lambda)^{s+1}}\right)\\
&= \Res{\lambda}\left(f^{[k]}(z)\sum_{s\geq 0}(-\partial_z)^s/s!
   \frac{N_{\lambda}^s}{(z-\lambda)}\right)
= \Res{\lambda}\left(\left[\sum_{s\geq0}\partial_z^s/s!
   f^{[k]}(z)N_{\lambda}^s\right] \frac{1}{z-\lambda}\right)\\ 
&= \Res{\lambda}\left(f^{[k]}(z+N_{\lambda})\frac
 1{z-\lambda}\right)=f^{[k]}(Z_{\lambda}).
\end{align*}
\end{proof}

Of course, in the above lemma $f^{[k]}(Z_{\lambda})$ is a matrix and may
not commute with other matrices appearing.  So, one needs a little
care in applying the Lemma~\ref{lem:matres}.

\section{Matrix KP Hierarchy}\label{sec:MKP}

Let $\V=(X,Z,A,B)\in\sCM{r}{n}$ and associate to it the \emph{wave
function} $\psi_{\V}$ depending on the spectral parameter $z$ and the
times $\t=(t_1,t_2,t_3,\ldots)$:
\begin{equation}
\label{psidef}
\psi_{\V}(\t,z)=\gamma(\t,z)\left(I_r+A\tilde X^{-1}(z I_n-Z)^{-1} B\right),
\end{equation}
where\footnote{Note that the dependence on $t_i$ in $-\tilde X$ is such that its eigenvalue dynamics are governed by the $i^{th}$ spin Calogero Hamiltonian.}  
$$
\tilde X=\tilde X(\t)=\sum_{i=1}^{\infty}i t_i Z^{i-1}-X,
\qquad\hbox{and}\qquad
\gamma(\t,z)=\exp\left(\sum_{i=1}^{\infty} t_i z^i\right).
$$


If $\p_{\V}(z)=\det(z I-Z)$ is the characteristic polynomial of $Z$
then the wave function \eqref{psidef} can be multiplied by $\p_{\V}$
and an exponential so as to yield a polynomial in $z$ with
coefficients that are rational in the times
\begin{equation}
K(\t,z)=\gamma^{-1}(\t,z)\psi_{\V}(\t,z)\p_{\V}(z).\end{equation}
Indeed, $q_{\V}(z)(zI-Z)^{-1}$ is the classical adjoint of
$zI_{n}-Z$ is and hence polynomial in $z$.  Letting $\partial=\frac{\partial}{\partial x}$ be the differential
operator in $x=t_1$, we note that the ordinary differential operator
$\K_{\V}=K(\t,\partial)$ satisfies
\begin{equation}
\psi_{\V}(\t,z)=\frac{1}{\p_{\V}(z)} \K_{\V} \gamma(\t,z).
\label{def:K}
\end{equation}

The main goal of this section is to prove that the pseudo-differential operator
$$
\calL_{\V}=\K_{\V}\circ \partial\circ \K_{\V}^{-1}
$$
is a solution to the matrix KP hierarchy in that it satisfies the Lax
equation \eqref{KPlax}.  As in \cite{MR1234841} (see also
\cite{MR1350415,MR87b:58039}), the proof will involve identifying finitely
supported distributions in $z$ that annihilate the function
$\psi_{\V}$.

\subsection{Conditions satisfied by $\psi_{\V}$}

Consider a generalized eigenvalue $\lambda$ of $Z$ with multiplicity
$\ell$ and use the notation of Section~\ref{sec:block} to denote by
$Z_{\lambda}$, $X_{\lambda\mu}$, $A_{\lambda}$, etc. the blocks of the
operators $X$, $Z$, $A$ and $B$.  Let $v\in\C^{r\ell+\ell}$ have the
decomposition
$$
v=\begin{pmatrix} v_0\cr v_1 \cr v_2 \cr \vdots\cr v_{\ell-1} \cr
    w\end{pmatrix}$$ 
where $ v_i\in \C^r$ and $w\in\C^\ell$ and define the
distribution $c_{v,\lambda}$ taking $r\times r$ matrix functions of
$z$ to $r$ component constant vectors by the formula
\begin{equation}
c_{v,\lambda}(f(z))=\Res{\lambda}\left(f(z)\cdot \left( A_{\lambda}(zI-Z_{\lambda})^{-1} w+\sum_{i=0}^{\ell-1} (z-\lambda)^i v_i \right)\right).
\label{eqn:condition}\end{equation}
In this section we will show that there are $r\ell$ linearly
independent distributions of this form satisfying
$c_{v,\lambda}(\psi_{\V}(\t,z))\equiv0$.  Consequently, by running
through all of the eigenvalues of $Z$ we obtain in this manner an
$rn$-dimensional space of conditions satisfied by the wave
function. Indeed, if $\{\lambda_{i}\}$ are the generalized eigenvalues
of the $n\times n$ matrix $Z$ with multiplicities $\{\ell_i\}$, then
$n=\sum \ell_{i}$.

Consider the $r\times (r\ell+\ell)$ matrix
\[
\Gamma_{\lambda}=\begin{pmatrix}B_{\lambda}& N_{\lambda}B_{\lambda}
&
N_{\lambda}^2 B_{\lambda}
&
\cdots
&
N_{\lambda}^{\ell-1} B_{\lambda}
&
-X_{\lambda\lambda}
\end{pmatrix}.
\]

\begin{claim}{Lemma}\label{lem:conditions}
If $v\in \ker\Gamma_{\lambda}$ 
then $c_{v,\lambda}\left(\psi_{\V}(\t,z)\right)=0$ for all values of the variables $\t$.  
\end{claim}
\begin{proof}
Note first that $\psi_{\V}(\t,z)$ has the block decomposition
\begin{equation}
\psi_{\V}(\t,z)=\gamma(z,t)\left(I-\sum_{\kappa,\mu} A_{\kappa}(\tilde X^{-1})_{\kappa\mu}(zI-Z_{\mu})^{-1} B_{\mu}\right)\label{eqn:blockpsi}
\end{equation}
where again the sum is taken over all (not necessarily distinct) pairs
of generalized eigenvalues $\kappa$ and $\mu$ of $Z$.

Now, we wish to use Lemma~\ref{lem:matres} to expand the residue in
\eqref{eqn:condition} where $f(z)$ is replaced by
\eqref{eqn:blockpsi}.  It will be convenient to introduce the
abbreviation
$$
C_{\mu}=-\sum_{\kappa}A_{\kappa}\tilde X^{-1}_{\kappa \mu},
$$
so that we have
\begin{equation}
\psi_{\V}(\t,z)=\gamma(\t,z)\left(I+\sum_{\mu} C_{\mu}(zI-Z_{\mu})^{-1}B_{\mu}\right)
\label{eqn:psiwC}
\end{equation}
and
\begin{equation}
A_{\lambda}=-\sum_{\mu}C_{\mu}\tilde X_{\mu\lambda}\label{CtoA}.
\end{equation}

The various contributions to the residue are usefully organized
according to dependence on $C_{\mu}$.  First of all, there is the
contribution independent of $C_{\mu}$.  It is given by
\begin{equation}
\Res{\lambda}\left(\gamma(\t,z)A_{\lambda}(zI-Z_{\lambda})^{-1}w\right)=A_{\lambda}\gamma(\t,Z_{\lambda})w.\label{eq:3}\tag{A}
\end{equation}
Making use of Lemma~\ref{lem:offdiagX} one finds that the
contributions containing $C_{\mu}$ for $\mu\not=\lambda$ are
\begin{equation}
  \sum_{\mu\not=\lambda}\Res{\lambda}\left((zI-Z_{\mu})^{-1}
    B_{\mu}A_{\lambda}(zI-Z_{\lambda})^{-1}w \gamma(\t,z)\right)
  =C_{\mu}\tilde X_{\mu\lambda}\gamma(\t,Z_{\lambda})w.\label{eq:4}\tag{B}
\end{equation}
 Next we turn to the terms involving $C_{\lambda}$.  The first one is
\begin{align}
&
C_{\lambda}\Res{\lambda}\left(\gamma(\t,z)(zI-Z_{\lambda})^{-1}B_{\lambda}
\sum_{i=0}^{\ell-1}(z-\lambda)^iv_i\right)\notag\\
&=
C_{\mu}\sum_{i,j=0}^{\ell-1}\Res{\lambda}\left(\gamma(\t,z)(zI-Z_{\lambda})^{i+1}
N_{\lambda}^i B_{\lambda}(z-\lambda)^jv_j\right)\notag\\
&=C_{\lambda}\gamma(\t,Z_{\lambda})\sum_{i=0}^{\ell-1}N_{\lambda}^i B_{\lambda}v_i.\label{eq:5}\tag{C}
\end{align}
The other term linear in $C_{\lambda}$ is
\begin{align}
  &C_{\lambda}\Res{\lambda}\left(\gamma(\t,z)(zI-Z_{\lambda})^{-1}
B_{\lambda}A_{\lambda}(zI-Z_{\lambda})^{-1}w\right)\notag\\
  =&
  C_{\lambda}\Res{\lambda}\left(\gamma(\t,z)(zI-Z_{\lambda})^{-1}
([X_{\lambda\lambda},Z_{\lambda}]-I)(zI-Z_{\lambda})^{-1}w\right)\notag\\
  =&-C_{\lambda}\Res{\lambda}\left(\gamma(\t,z)(zI-Z_{\lambda})^{-2}w\right)\notag\\
  &+ C_{\lambda} \Res{\lambda}\left(\gamma(\t,z)(zI-Z_{\lambda})^{-1}
([X_{\lambda\lambda},Z_{\lambda}]-z)(zI-Z_{\lambda})^{-1}w\right)\notag\\
  =& -C_{\lambda}\gamma'(\t,Z_{\lambda})w-C_{\lambda}\Res{\lambda}
(\gamma(\t,z)(zI-Z_{\lambda})^{-1}X_{\lambda\lambda}w)\notag\\
  &+ C_{\lambda} \Res{\lambda}(X_{\lambda\lambda}(zI-Z_{\lambda})^{-1}\gamma(\t,z)w)\notag\\
  =& -C_{\lambda}\gamma'(\t,Z_{\lambda})w-C_{\lambda}
\gamma(\t,Z_{\lambda})X_{\lambda\lambda}w+C_{\lambda}X_{\lambda\lambda}\gamma(\t,Z_{\lambda})w\notag\\
  =&C_{\lambda} \tilde
  X_{\lambda\lambda}\gamma(\t,Z_{\lambda})w-C_{\lambda}\gamma(\t,Z_{\lambda})X_{\lambda\lambda}w.
\tag{D}\label{eq:6}
\end{align}
Since $v\in\ker\Gamma_{\lambda}$ is equivalent to the statement
\begin{equation}
\sum_{k=0}^{\ell-1}N_{\lambda}^kB_{\lambda}v_k-X_{\lambda\lambda}w=0,
\label{eqn:kergam}
\end{equation}
we see that $(C)$ cancels against the second term in $(D)$.  So, combining all four terms gives
$$
(A)+(B)+(C)+(D)=A_{\lambda}\gamma(\t,Z_{\lambda})w+\sum_{\mu\not=\lambda}C_{\mu}\tilde X_{\mu\lambda}\gamma(\t,Z_{\lambda})w+C_{\lambda}\tilde X_{\lambda\lambda}\gamma(\t,Z_{\lambda})w.
$$
Applying \eqref{CtoA} shows that this is equal to zero as required.
\end{proof}

\begin{claim}{Lemma}\label{lem:dimensions}
The distributions $c_{v,\lambda}$ for $v\in\ker \Gamma_{\lambda}$ form an $r\ell$-dimensional space.
\end{claim}
\begin{proof}
Note that the map $\Omega:v\in\ker \Gamma_{\lambda}\mapsto c_{v,\lambda}$ is itself a linear map.   What we need to prove, therefore is that
$$
\dim\ker\Gamma_{\lambda}-\dim\ker\Omega=r\ell.
$$


A vector $v$ clearly does not lie in the kernel of $\Omega$ if
$v_i\not=0$ for any $i$.  The dimension of the kernel of $\Omega$ is
therefore equal to the dimension of the space of vectors $w$ with the
property that $X_{\lambda\lambda}w=0$ and
$A_{\lambda}(zI-Z_{\lambda})^{-1}w=0$.
In fact, we will show that the only such $w$ is the zero vector (and
hence that $\dim\ker\Omega=0$).  Beginning with the fact that
$$
[(zI-Z_{\lambda}),X_{\lambda\lambda}]-B_{\lambda}A_{\lambda}=I.
$$
Multiplying by $(zI-Z_{\lambda})^{-1}$ on the right, applying both
sides of the resulting equation to $w$ and then multiplying by
$(zI-Z_{\lambda})^{-1}$ on the left gives us that
$$
X_{\lambda\lambda}(zI-Z_{\lambda})^{-1}w=(zI-Z_{\lambda})^{-2}w.
$$
Expanding both sides of this equation in terms of powers of
$(z-\lambda)$ and equating like powers gives us that
$$
X_{\lambda\lambda}N_{\lambda}^kw=kN_{\lambda}^{k-1}w,\quad\text{for $k>0$}.
$$
Since $N_{\lambda}$ is nilpotent, for a sufficiently large $k$ the
left-hand side is equal to zero.  But the equation then tells us that
$N_{\lambda}^{k-1}w$ is then also equal to zero, which again means
that the left hand side would be zero for a smaller value of $k$.
Repeating this process until $k=1$ we find that $w=0$.

A similar argument shows that $\dim\ker\Gamma_{\lambda}=r\ell$.
Considering instead the vectors $w$ such that
$w^{\top}\Gamma_{\lambda}=0$ implies that
$w^{\top}X=w^{\top}(zI-Z_{\lambda})^{-1}B_{\lambda}=0$ and the same
process reveals that $w=0$ so that $\Gamma_{\lambda}$ has rank $\ell$.
Consequently, its kernel has dimension $(r\ell+\ell)-\ell=r\ell$.
\end{proof}

\subsection{The Kernel of $\K_{\V}$}

The results of the previous section on the distributions annihilating
$\psi_{\V}$ give us information about the kernel of the matrix
ordinary differential operator $\K_{\V}$ defined in \eqref{def:K}:

\claim{Corollary}\label{cor1} Let $c_{v,\lambda}$ be as in
Lemma~\ref{lem:conditions}.  Then the  $r$ component vector valued function
\begin{equation}
\phi_{v,\lambda}(\t)=c_{v,\lambda}\left(\frac{\gamma(\t,z)}{\p_{\V}(z)}\right)\label{phiform}
\end{equation}
is in the kernel of the operator $\K_{\V}$. 

\begin{proof}
Since $c_{v,\lambda}$ commutes with multiplication and differentiation in $x=t_1$, we have
$$
\K_{\V}\phi_{v,\lambda}=c_{v,\lambda}\left(\K_{\V} \gamma(\t,z)q_{\V}^{-1}(z)\right)=c_{v,\lambda}(\psi_{\V})=0,
$$
by \eqref{psidef} and Lemma \ref{lem:conditions}.
\end{proof}

In fact, the entire kernel of $\K_{\V}$ is spanned by functions of
this form, and as a consequence they satisfy certain useful linear
differential equations.

\claim{Corollary}\label{vaccor} If $\phi(\t)$ is a vector in the
kernel of $\K_{\V}$ then it is a linear combination of the
$\phi_{v,\lambda}(\t)$, and so satisfies the equation
$$
\frac{\partial^k}{\partial t_1^k} \phi(\t)=\frac{\partial}{\partial t_k}\phi(\t).
$$
\begin{proof}
  By making use of all of the eigenvalues of $Z$, Corollary~\ref{cor1}
  gives us $rn$ linearly independent vector functions in the kernel of
  the $n$th order ordinary differential operator with $r\times r$
  matrix coefficients.  Since they are linearly independent (those
  corresponding to the same eigenvalue are linearly independent by
  Lemma~\ref{lem:dimensions} and those corresponding to different
  eigenvalues cannot be linearly dependent due to the factor of
  $e^{x\lambda}$) this accounts for the entire kernel of $\K_{\V}$.

Note that $\gamma(\t,z)$ trivially satisfies 
$$
\frac{\partial^k}{\partial t_1^k} \gamma(\t,z)=z^k\gamma(\t,z)=\frac{\partial}{\partial t_k}\gamma(\t,z).
$$
Now, the proof here is elementary because differentiation in $t_i$ commutes with the residue, multiplication by functions of $z$ and matrix multiplication in the definition of $\phi_{v,\lambda}$ and applies by linearity to the entire kernel.
\end{proof}

\subsection{The Lax Equation} 
Now we come to main point of this section. If the 
$\V$ moves according the spin Calogero dynamics the wave function
$\psi_{\V}$ depends on the time variables $\t$, and this produces a
solution of the matrix KP hierarchy. More precisely:

\claim{Theorem}\label{thm:lax} The pseudo-differential operator $\calL_{\V}=\K_{\V}\circ \partial \circ \K_{\V}^{-1}$ satisfies
$$
\frac{\partial}{\partial t_i}\calL_{\V} =[(\calL_{\V}^i)_+,\calL_{\V}].
$$

\begin{proof}
  First, we note that the (pseudo)-differential operator
  $(\calL_{\V}^i)_-\circ \K_{\V}$ is actually a 
  \emph{differential} operator since
$$
(\calL_{\V}^i)_-\circ \K_{\V} + (\calL_{\V}^i)_+\circ \K_{\V} = \K_{\V}\circ \partial^i
$$
and therefore $(\calL_{\V}^i)_-\circ \K_{\V} =-(\calL_{\V}^i)_+\circ
\K_{\V} +\K_{\V}\circ \partial^i$.

Now, let $\phi(x)$ be a vector function in the kernel of the operator
$\K_{\V}$.  Then applying $\frac{\partial}{\partial t_i}$ to the
equality $\K_{\V}\phi=0$ and using Corollary~\ref{vaccor} we find
\begin{eqnarray*}
0&=& \frac{\partial}{\partial t_i}\circ \K_{\V}(\phi)=(\K_{\V})_{t_i}(\phi)+\K_{\V}(\phi_{t_i})\\
&=&( \K_{\V})_{t_i}(\phi)+\K_{\V} (\partial^i \phi)\\
&=& (\K_{\V})_{t_i}\phi+\calL_{\V}^i\circ\K_{\V}(\phi)\\
&=& (\K_{\V})_{t_i}(\phi)+(\calL_{\V}^i)_+\circ \K_{\V}(\phi)+(\calL_{\V}^i)_-\circ \K_{\V}(\phi)
\end{eqnarray*}
However, since $\phi$ is in the kernel of $\K_{\V}$ we know that
$(\calL_{\V}^i)_+\circ \K_{\V}(\phi)=0$.  Then the last displayed
equality gives us that the entire kernel of $\K_{\V}$ is in the kernel
of the ordinary differential operator
$(\K_{\V})_{t_i}+(\calL_{\V}^i)_-\circ \K_{\V}$.  Since this operator
has order strictly less than $n$, it can only have such a large kernel
if it is the zero operator and we conclude
$$
(\K_{\V})_{t_i}=-(\calL^i)_-\circ \K_{\V}.
$$

Using this we find that
\begin{eqnarray*}
(\calL_{\V})_{t_i}&=& (\K_{\V})_{t_i}\circ \partial \circ \K_{\V}^{-1}-\K_{\V}\circ \partial \circ \K_{\V}^{-1}\circ (\K_{\V})_{t_i}\circ K^{-1}\\
&=& -(\calL_{\V}^i)_- \circ \K_{\V} \circ \partial \circ \K_{\V}^{-1}+\K_{\V}\circ \partial\circ \K_{\V}^{-1}\circ (\calL_{\V}^i)_- \circ \K_{\V} \circ \K_{\V}^{-1}\\
&=& [\calL_{\V},(\calL_{\V}^i)_-]=[(\calL_{\V}^i)_+,\calL_{\V}].
\end{eqnarray*}

\end{proof}

\section{Symmetries}\label{sec:symmetries}

The symmetry $X\mapsto X+cZ^j$ of $\sCM{r}{n}$ induces the integrable
dynamics of both the particle system and the wave equations of the
matrix KP hierarchy.  Here are some other symmetries and how they
affect the KP Lax operator.

\subsection{Action of $\GL(n)$}\label{sec:GLn}

For any $G\in\GL(n)$ we define $S_G$ acting on $\sCM{r}{n}$ by
\begin{eqnarray*}
S_G: \sCM{r}{n}&\to&\sCM{r}{n}\\
(X,Z,A,B)&\mapsto& (GXG^{-1},GZG^{-1},AG^{-1},GB).
\end{eqnarray*}
The Lax operator $\calL_{\V}$ is unaffected by this action: $\calL_{S_G(\V)}=\calL_{\V}$ .

Since this symmetry does not affect the corresponding dynamical
objects from the previous sections, it makes sense to consider
$\sCM{r}{n}$ \textit{modulo} this group action as the phase space of
the spin Calogero particle dynamics as well as the corresponding
matrix KP solutions.

\subsection{Action of $\GL(r)$}

Clearly, if we have $G\in\GL(r)$ then we can conjugate solutions to
the matrix KP hierarchy to get solutions that are technically
different, but not very different.  This also manifests itself as a
group action on the level of $\sCM{r}{n}$.  Let $G\in\GL(r)$ then if
\begin{eqnarray*}
s_G:\sCM{r}{n}&\to&\sCM{r}{n}\\
(X,Z,A,B)&\mapsto&(X,Z,GA,BG^{-1})
\end{eqnarray*}
one finds that $\calL_{s_G(\V)}=G\calL_{\V}G^{-1}$.

\subsection{Changing $r$}

There is an easy way to take an $r\times r$ solution and turn it into
an $R\times R$ solution for $r<R$. 
Let $a$ be an $R\times r$ and $b$ an $r\times R$ matrix such that
$ba=I$ is the $r\times r$ identity matrix.  Defining
$U_{a,b}:\sCM{r}{n}\to\sCM{R}{n}$ by
$$
U_{a,b}(X,Z,A,B)=(X,Z,aA,Bb).
$$
one then has
$$
\calL_{U_{a,b}(\V)}(\t)=a\calL_{\V}(\vec t)b.
$$

\subsection{The Bispectral Involution}\label{sec:bisp}

Finally, we have an important discrete symmetry whose effect on the KP solution
will be the subject of the next section:
\begin{eqnarray*}
\b:\sCM{r}{n}&\to&\sCM{r}{n}\\
\V=(X,Z,A,B)&\mapsto&\V^{\b}=(Z^{\top},X^{\top},B^{\top},A^{\top})
\end{eqnarray*}

The significance of this symmetry on the KP solution is most easily
seen by looking at the wave function $\psi_{\V}$ as a function of
$x=t_1$ and $z$ only (setting all of the other times equal to zero).
As in the case $r=1$, it involves an exchange of $x$ and $z$, but when
$r>1$ one must also take the transpose of the function:
\begin{equation}
\psi_{\V}(x,z)=\psi_{\V^{\b}}^{\top}(z,x).
\label{xz-switch}
\end{equation}

\section{Bispectrality}\label{sec:bispectrality}

Let $\V\in\sCM{r}{n}$ and $\p_{\V}(z)=\det(zI-Z)$.  In this section,
since the dynamics are not significant, we will consider $x=t_1$ and
$t_i=0$ for $i>1$.  Thus, for instance, we will write $\psi_{\V}(x,z)$
for $\psi_{\V}((x,0,0,\ldots),z)$ and $ \gamma(x,z)=e^{xz}.  $
In the next section we will associate two commutative rings of
ordinary differential operators (one acting in $x$ and one acting in
$z$) to the choice of $\V\in\sCM{r}{n}$ and then in the following
section we will demonstrate that $\psi_{\V}(x,z)$ is a common
eigenfunction for the operators in the rings.  In particular, any operator from each of the rings along with $\psi_{\V}$ form a bispectral triple as in Definition~\ref{def:bisp}.

\subsection{Commutative Rings of Matrix Differential Operators}

Associate to $\V\in\sCM{r}{n}$ the ring $R_{\V}\subset \C[z]$ defined
by the property that the polynomials preserve the conditions
annihilating $\psi_{\V}(\t,z)$ from Lemma~\ref{lem:conditions}:
$$
R_{\V}=\left\{p\in\C[z]\,|\,
c_{v,\lambda}\left(\psi_{\V}\right)=0\Rightarrow c_{v,\lambda}\left(p(z)\psi_{\V}\right)=0\right\}.
$$

\claim{Lemma}\label{lem:nontrivring} The ring $R_{\V}$ is non-empty.  In particular,
$$
\p_{\V}^2(z)\C[z]\subset R_{\V}.
$$
\begin{proof}
  Note that $\p_{\V}(z)\psi_{\V}(\t,z)=\K_{\V}\gamma(\t,z)$ is
  non-singular in $z$.  Then the claim follows from the fact that
  $c_{v,\lambda}(\p_{\V}(z) f(z))=0$ for any non-singular function
  $f$.
\end{proof}
 
By substituting the pseudo-differential operator $\calL_{\V}$ into
these polynomials, we associate a commutative ring of
pseudo-differential operators 
$$\R_{\V}=\left\{p(\calL_{\V})|p\in
  R_{\V}\right\}$$ 
to $\V$.  However, as the next lemma demonstrates,
these are in fact differential operators.

\claim{Lemma}\label{odolem}   
If $p\in R_{\V}$ then $L=p(\calL_{\V})$ is a differential operator (as opposed to a general pseudo-differential operator) satisfying the eigenvalue equation
$$
L\psi_{\V}(x,z)=p(z)\psi_{\V}(x,z).
$$

\begin{proof}
  Since the leading coefficient of $\K_{\V}$ is a nonsingular matrix
  (in fact, it is the identity matrix because of the form of
  $\psi_{\V}$), it is sufficient to show that the kernel of $\K_{\V}$
  is contained in the kernel of $\K_{\V}\circ p(\partial)$ because
  then we know that this ordinary differential operator factors as
  $L\circ \K_{\V}$ for some ordinary differential operator $L$ which
  meets all of the other criteria.

  So, now let $\phi(x)$ be a function in the kernel of $\K_{\V}$.  By
  Lemma~\ref{cor1} we know that $\phi(x)$ is a linear combination of
  functions of the form \eqref{phiform}.  However,
$$
\K_{\V}\circ p(\partial)c_{v,\lambda}(\gamma(x,z)q_{\V}^{-1}(z))
=
c_{v,\lambda}(p(z)\psi_{\V}(x,z))
=0
$$
by \eqref{def:K} and the definition of $R_{\V}$.  
\end{proof}

We will also associate a commutative ring of ordinary differential
operators in $z$ to $\V$.  Applying the procedure above to the point
$\V^{\b}=(Z^{\top},X^{\top},B^{\top},A^{\top})\in\sCM{r}{n}$ we have
another commutative ring $\R_{\V^{\b}}$ of differential operators in
$x$.  We convert them to differential operators in $z$ by simply
replacing $x$ with $z$, $\partial_x$ with $\partial_z$ and transposing
the coefficients:
$$
\R^{\b}_{\V}=\left\{ L^{\top}(z,\partial_z) | L(x,\partial_x)\in\R_{\V^{\b}}\right\}.
$$

\subsection{Common Eigenfunction}

Our main result is the observation that $\psi_{\V}(x,z)$ is a common
eigenfunction for the differential operators in the rings $\R_{\V}$
and $\R^{\b}_{\V}$ satisfying eigenvalue equations of the form
\eqref{bispdef2}: \claim{Theorem} Let $p\in R_{\V}$ and $\pi\in
R_{\V^{\b}}$, then there exist ordinary differential operators
$L(x,\partial_x)\in\R_{\V}$ and $\Lambda(z,\partial_z)\in\R^{\b}_{\V}$
such that
$$
\relax{L}\psi_{\V}(x,z)=p(z) \psi_{\V}(x,z)\qquad
\hbox{and}
\qquad
\Right{\Lambda}\psi_{\V}(x,z)=\pi(x)\psi_{\V}(x,z).
$$
\begin{proof}
The first equation follows from Lemma~\ref{odolem}.  Similarly, it follows from Lemma~\ref{odolem} that
there is a differential operator $Q(x,\partial_x)$ with the property that
$$
Q\psi_{\V^{\b}}(x,z)=\pi(z)\psi_{\V^{\b}}(x,z).
$$
Exchanging the roles of $x$ and $z$ in this equation, taking the
transpose (see ), and applying \eqref{eq:RightTranspose} and
\eqref{xz-switch} results in the second equation of the claim.

 \end{proof}
 
\section{Example}

For the sake of clarity, we briefly illustrate the main ideas with an
example.

Consider $\V=(X,Z,A,B)\in\sCM{2}{3}$ where
$$
X=\left(
\begin{matrix}
 0 & 0 & 0 \cr
 -1 & 0 & -1 \cr
 1 & 0 & 2
\end{matrix}
\right)
\qquad
Z=\left(
\begin{matrix}
 0 & 1 & 0 \cr
 0 & 0 & 0 \cr
 0 & 0 & 0
\end{matrix}
\right)
$$
$$
A=
\left(
\begin{matrix}
0 & 1 & 0 \cr
 0 & 0 & 1
\end{matrix}
\right)
\qquad
B=\left(
\begin{matrix}
 0 & 1 \cr
 -2 & 0 \cr
 1 & -1
\end{matrix}
\right).
$$
Then
$$
\psi_{\V}(x,z)=e^{xz}\left(I + 
\left(
\begin{matrix}
 \frac{-2 z x^2+3 z x+2 x-2}{(x-2) x^2 z^2} & \frac{1}{(x-2) x^2 z} \cr
 \frac{x z-2}{(x-2) x z^2} & \frac{1-x}{(x-2) x z}
\end{matrix}
\right)
\right).
$$
This can be written as
$\psi_{\V}=\K_{\V}e^{xz}/\p_{\V}(z)$ where
$
\p_{\V}(z)=\det(z I - Z)=z^3
$
and
$$
\K=
\left(
\begin{matrix}
 \frac{2 (x-1)}{(x-2) x^2} & 0 \cr
 -\frac{2}{(x-2) x} & 0
\end{matrix}
\right)\partial+\left(
\begin{matrix}
 \frac{3-2 x}{(x-2) x} & \frac{1}{(x-2) x^2} \cr
 \frac{1}{x-2} & \frac{1-x}{(x-2) x}
\end{matrix}
\right)\partial^2+\left(
\begin{matrix}
 1 & 0 \cr
 0 & 1
\end{matrix}
\right)\partial^3
$$
is an ordinary differential operator.

To find the conditions satisfied by $\psi_{\V}(x,z)$, we note that
the kernel of 
$$
\Gamma_0=
\left(
\begin{matrix}
 0 & 1 & -2 & 0 & 0 & 0 & 0 & 0 & 0 \cr
 -2 & 0 & 0 & 0 & 0 & 0 & 1 & 0 & 1 \cr
 1 & -1 & 0 & 0 & 0 & 0 & -1 & 0 & -2
\end{matrix}
\right)
$$
is made up of vectors of the form
$$
v=\left(
\begin{matrix}
 \cc(1) & \cc(2) & \frac{\cc(2)}{2} & \cc(4) &
 \cc(5) & \cc(6) &
   3 \cc(1)+\cc(2) & \cc(8) &
 -\cc(1)-\cc(2)
\end{matrix}
\right)^{\top}.
$$
Hence, we conclude that $c_{v,0}(\psi_{\V}(x,z))=0$ where
$$
c_{v,0}(f(z))=\Res{0}\left(f(z)\left(
\begin{matrix}
 \cc(5) z^2+\frac{1}{2} \cc(2)
 z+\cc(1)+\frac{\cc(8)}{z} \cr \cc(6)
   z^2+\cc(4)
 z+\cc(2)+\frac{-\cc(1)-\cc(2)}{z}
\end{matrix}
\right)\right)
.
$$

Moreover, $c_{v,0}(p(z)\psi_{\V})=0$ whenever $p\in z^4\C[z]=R_{\V}$.
Consequently, we can find an ordinary differential operator having any
of these polynomials as its eigenvalue.  In particular, solving
$$
\K_{\V}\circ \partial^4= L \circ \K_{\V}
$$
for $L$ we find
\begin{eqnarray*}
L &=& 
\left( \begin{matrix}
 -\frac{8 \left(3 x^4-18 x^3+50 x^2-63 x+30\right)}{(x-2)^4 x^4} & \
\frac{4 \left(23 x^3-93 x^2+138 x-72\right)}{(x-2)^4 x^5} \cr
 \frac{8 \left(2 x^3-5 x^2+9 x-6\right)}{(x-2)^4 x^3} & -\frac{4 \
\left(2 x^4-8 x^3+27 x^2-42 x+24\right)}{(x-2)^4 x^4}
\end{matrix}
\right)\\&&
+\left(
\begin{matrix}
 \frac{4 \left(6 x^3-27 x^2+49 x-30\right)}{(x-2)^3 x^3} & \
-\frac{4 \left(13 x^2-35 x+26\right)}{(x-2)^3 x^4} \cr
 -\frac{4 \left(4 x^2-7 x+6\right)}{(x-2)^3 x^2} & \frac{4 \
\left(2 x^3-6 x^2+13 x-10\right)}{(x-2)^3 x^3}
\end{matrix}
\right)\partial
\\
&&+\left(
\begin{matrix}
 -\frac{8 \left(x^2-3 x+3\right)}{(x-2)^2 x^2} & \frac{4 (3 \
x-4)}{(x-2)^2 x^3} \cr
 \frac{4}{(x-2)^2} & -\frac{4 \left(x^2-2 x+2\right)}{(x-2)^2 \
x^2}
\end{matrix}
\right)\partial^2+\partial^4
\end{eqnarray*}
which satisfies
$
L\psi_{\V}=z^4\psi_{\V}$.

Of course, we can follow this same procedure beginning with another
element of $\sCM{2}{3}$.  In particular, if we begin with
$$
\V^{\b}=(Z^{\top},X^{\top},B^{\top},A^{\top})
$$
instead then the differential operator we produce will be
\begin{eqnarray*}
Q&=&\left(
\begin{matrix}
 -\frac{4 \left(16 x^2+65 x+90\right)}{x^6} & \frac{80 \
x+216}{x^6} \cr
 \frac{8 x+12}{x^4} & -\frac{4 \left(4 x^2+21 x+36\right)}{x^6}
\end{matrix}
\right)+\left(
\begin{matrix}
 \frac{4 \left(16 x^2+65 x+90\right)}{x^5} & -\frac{8 (10 \
x+27)}{x^5} \cr
 -\frac{4 (2 x+3)}{x^3} & \frac{4 \left(4 x^2+21 \
x+36\right)}{x^5}
\end{matrix}
\right)\partial\\&&+\left(
\begin{matrix}
 -\frac{2 \left(12 x^2+65 x+90\right)}{x^4} & \frac{40 \
x+108}{x^4} \cr
 \frac{6}{x^2} & -\frac{8 x^2+42 x+72}{x^4}
\end{matrix}
\right)\partial^2+\left(
\begin{matrix}
 \frac{34 x+60}{x^3} & -\frac{4 (2 x+9)}{x^3} \cr
 0 & \frac{14 x+24}{x^3}
\end{matrix}
\right)\partial^3\\&&+\left(
\begin{matrix}
 4-\frac{12}{x^2} & \frac{6}{x^2} \cr
 0 & 4-\frac{6}{x^2}
\end{matrix}
\right)\partial^4-4\partial^5+\partial^6.
\end{eqnarray*}
The function $\psi_{\V^{\b}}$ is an eigenfunction for this operator
satisfying $Q\psi_{\V^{\b}}(x,z)=(4z^4-4z^5+z^6)\psi_{\V^{\b}}(x,z)$.

More interestingly, since $\psi_{\V^{\b}}(x,z)=\psi_{\V}^{\top}(z,x)$,
if we transpose the matrix coefficients on $Q$, replace $x$ with $z$
and $\partial=\partial_x$ with $\partial_z$, we get a differential
operator $\Lambda$ in the variable $z$.  This operator applied to
$\psi_{\V}(x,z)$ (the wave function computed earlier) \textit{from the
  right} satisfies
$$
\Right{\Lambda}\psi_{\V}(x,z)=(4x^4-4x^5+x^6)\psi_{\V}(x,z),
$$
demonstrating bispectrality.

\section{Conclusions and Comments}

The main results of the present paper can be viewed as another step in
addressing the ``bispectral problem'' of F.A. Gr\"unbaum
\cite{MR1611018,MR826863}, seeking operators satisfying eigenvalue equations
of the form \eqref{bispdef1}.  In \cite{MR2148533}, the authors considered
the case in which one of the operators is a second order difference
operator with matrix coefficients and the two operators act on matrix
eigenfunctions from different directions.  However, bispectrality for
matrix differential operators has only been studied with both
operators acting from the left \cite{MR1067913}.  Here we consider the case
\eqref{bispdef2} in which the operators are $r\times r$ matrix
differential operators acting from different directions.  Since our
construction conveniently reproduces the results of Wilson's seminal
paper \cite{MR1234841} in the special case $r=1$, this particular formulation
of the bispectral problem appears to be the correct one for
generalizing those results to the case of matrix differential
operators.  However, the method of proof and especially the explicit
formulation of the ``conditions'' satisfied by the wave function above
are novel even for $r=1$.

In \cite{MR1234841}, it was shown that the bispectral operators associated to
$\sCM{1}{n}$ are in fact the \textit{only} bispectral scalar ordinary
differential operators which commute with operators of relatively
prime order up to obvious renormalizations and changes of
variable\footnote{These were called ``rank one'' operators in that
  context, but we will avoid that terminology here to avoid confusion
  with the rank $r$ which is something different.}.  By
Lemma~\ref{lem:nontrivring} it follows that the differential operators
produced by the construction in this paper also all have the property
that they commute with other differential operators of relatively
prime order.

In addition, this paper can be seen as contributing to the literature
establishing a link between bispectrality and duality in classical and
quantum integrable systems.  (See, for instance,
\cite{MR1780352,MR2280341,MR1350415,MR1843573,MR1862956,MR1626461}.)
Again, the main results of the present paper for the spin Calogero
system in the case $r=1$ reproduce results previously presented for
the spinless case in \cite{MR1350415,MR1626461}.

Some questions arise naturally which we have not pursued.  There are
additional commuting Hamiltonians for the spin Calogero system
\cite{MR761664} and corresponding isospectral deformations for the
multi-component KP hierarchy \cite{MR97c:58061}, but their relationship to
bispectrality has not been explored here.  We have not looked at the
algebro-geometric implications of the rings $\R_{\V}$.  Certainly as
in the case $r=1$ \cite{MR1234841,MR1234841}, these contain operators of relatively
prime order are isomorphic to the coordinate rings of rational curves
with only cuspidal singularities.  However, whether there is any
further algebro-geometric significance such as was found in
\cite{MR1904791} or whether every commutative ring of matrix ordinary
differential operators with these properties is necessarily bispectral
have not been considered.  These questions, along with the obvious
question of what other matrix differential operators satisfy equations
of the form \eqref{bispdef2} will hopefully be addressed in future
papers.



\def\MR#1{\relax}
\par\medskip\noindent\textbf{Acknowledgments:} The second author was partially supported by
the NSF Grant DMS-0400484.  The third author appreciates helpful discussions with Tom Ivey, Folkert~M\"uller-Hoissen, and Oleg~Smirnov and the support of his department during his sabbatical.

\end{document}